\documentclass [a4paper,12pts] {article}
\usepackage {a4wide}
\usepackage {amssymb}
\usepackage {amsmath}
\usepackage {multirow, array}
\usepackage [english] {babel}
\usepackage {amsthm}
\usepackage {enumerate}
\usepackage[utf8]{inputenc}
\usepackage {tikz}

\newtheorem{dfn}{Definition}
\newtheorem{theo}[dfn]{Theorem}
\newtheorem{prop}[dfn]{Proposition}

\newtheorem{lem}[dfn]{Lemma}

\newtheorem*{quest}{Question}

\newcommand{\R}{\mathcal{R}}
\renewcommand{\L}{\mathcal{L}}
\newcommand{\N}{\mathcal{N}}
\renewcommand{\P}{\mathcal{P}}
\newcommand{\UU}{\mathcal{U}}

\newcommand{\MM}{\mathcal{M}}
\newcommand{\OO}{\mathcal{O}}
\newcommand{\QQ}{\mathcal{Q}}
\newcommand{\ZZ}{\mathbb{Z}}
\newcommand{\NN}{\mathbb{N}}
\newcommand{\delete}[1]{}
\newcommand{\GR}{G^{\boldsymbol{R}}}
\newcommand{\GL}{G^{\boldsymbol{L}}}
\newcommand{\HR}{H^{\boldsymbol{R}}}
\newcommand{\HL}{H^{\boldsymbol{L}}}

\newcommand{\Usum}{{\UU_1+\UU_2}}
\newcommand{\cl}{{c\ell}}
\newcommand{\cla}{{\cl(a, \overline{a})}}
\newcommand{\mz}{{\mathcal{M}_{\mathbb{Z}}}}
\newcommand{\qz}{{\mathcal{Q}_{\mathbb{Z}}}}

\title{Invertibility modulo dead-ending no-$\mathcal{P}$-universes}

\author{Gabriel Renault\\
\ \\
University of Mons - UMONS, Place du Parc 20, 7000 Mons, Belgium}

\begin{document}

\maketitle

\begin{abstract}
In normal version of combinatorial game theory, all games are invertible, whereas only the empty game is invertible in misère version.
For this reason, several restricted universes were earlier considered for their study, in which more games are invertible.
We here study combinatorial games in misère version, in particular universes where no player would like to pass their turn
In these universes, we prove that having one extra condition makes all games become invertible.
We then focus our attention on a specific quotient, called $\qz$, and show that all sums of universes whose quotient is $\qz$ also have $\qz$ as their quotient.
\end{abstract}

\begin{sloppypar}

\section{Introduction}

A combinatorial game is a finite two-player game with no chance and perfect information.
The players, called Left and Right\footnote{By convention, Left is a female player whereas Right is a male player.}, alternate moves until one player is unable to move.
The last player to move loses the game under the mis\`ere convention, while that same player would win under normal convention.

The conditions that make a game combinatorial ensure that one of the player has a winning strategy.
The main objective of combinatorial game theory is to determine which player should win and what their strategy is.
A basic way would be to look at all possible moves for both players all the way until the game ends in all branches and backtrack the winning player up to the original position.
Unfortunately, this method is usually quite time-consuming and often space-consuming as well.
Hence other approaches were developped, some specific to particular games and some more general.
One general approach is to decompose the position into a sum of smaller positions, study them separately and conclude on their sum.
It is thus interesting to be able to simplify the smaller positions before looking at the larger picture, including intermediate sums.

Finding invertibility of games is one of the most efficient ways to simplify sums of games.
It enables a simplification to the zero game, that is the game with no move.
Under normal convention, all games are invertible, whereas in mis\`ere version, the only invertible game is the empty game.
Mis\`ere games were thus studied in a more restrictive context\cite{mpg, dicotaxo, sprigs, pkayles, alternating, deadending, quotient1, quotient2, wothese}, where more games are invertible.
In some cases, all games are invertible\cite{sprigs, pkayles, deadending, wothese}.
This especially happens in all contexts studied so far where no player would ever want to pass their turn\cite{deadending,wothese}.
Hence it is natural to wonder if it is always the case.

\subsection{Preliminaries}

A game can be defined recursively by its sets of options $G = \{\GL|\GR\}$, where $\GL$ is the set of games Left can reach in one move (called Left options), and $\GR$ the set of games Right can reach in one move (called Right options).
The typical Left option of $G$ is denoted $G^L$, and the typical Right option of $G$ is denoted $G^R$.
A follower of a game $G$ is a game that can be reached from $G$ after a succession of (not necessarily alternating) Left and Right moves.
Note that a game $G$ is considered one of its own followers.
The zero game $0=\{\cdot|\cdot\}$, is the game with no options (a dot indicates an empty set of options).
A Left end (resp. Right end) is a game where Left (resp. Right) cannot move.

The (disjunctive) sum $G + H$ of two games $G$ and $H$ is defined recursively as \mbox{$G+H = \{\GL+H,G+\HL|\GR+H,G+\HR\}$}, where $\GL+H$ is understood to range over all sums of $H$ with an element of $\GL$, that is the game where each player can on their turn play a legal move for them in one (but not both) of the components.
The conjugate $\overline{G}$ of a game $G$ is recursively defined as $\overline{G} = \{\overline{\GR}|\overline{\GL}\}$, where again $\overline{\GR}$ is understood to range over all conjugates of elements of $\GR$, that is the game where Left's and Right's roles are reversed.

A game can also be depicted by its game tree, where the game trees of its options are linked to the root by downward edges, left-slanted for Left options and right-slanted for Right options.
It can be more readable than the bracket notation.
For instance, the game trees of a few games are depicted on Figure~\ref{fig:trees} with their bracket notations under the trees.

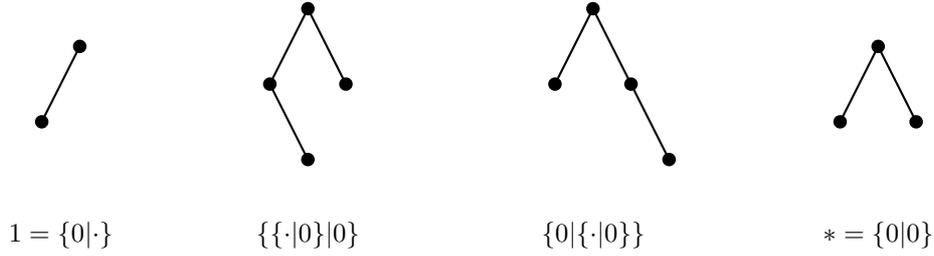
\begin{figure}
\begin{center}

\begin{tikzpicture}
[thick,scale=1,
     vertex/.style={circle,draw=white!100,inner sep=1pt,minimum
size=2mm,fill=white!100},
     blackvertex/.style={circle,draw,inner sep=0pt,minimum
size=1.5mm,fill=black!100},
     clause/.style={circle,draw,inner sep=0pt,minimum
size=3mm,fill=white!100}]


\coordinate (z) at (-5,1.5);
\coordinate (y) at (-5.5,0.5);

\coordinate (x) at (5.5,1.5);
\coordinate (w) at (5,0.5);
\coordinate (v) at (6,0.5);

\coordinate (a) at (-2,2);
\coordinate (b) at (-2.5,1);
\coordinate (c) at (-2,0);
\coordinate (d) at (-1.5,1);

\coordinate (e) at (1.75,2);
\coordinate (f) at (1.25,1);
\coordinate (g) at (2.25,1);
\coordinate (h) at (2.75,0);

\coordinate (1) at (-2,-1);
\coordinate (2) at (1.75,-1);
\coordinate (3) at (-5.25,-1);
\coordinate (4) at (5.5,-1);

\draw (c)--(b)--(a)--(d);
\draw (f)--(e)--(g)--(h);
\draw (y)--(z);
\draw (v)--(x)--(w);


\draw (a) node[blackvertex] {};
\draw (b) node[blackvertex] {};
\draw (c) node[blackvertex] {};
\draw (d) node[blackvertex] {};

\draw (e) node[blackvertex] {};
\draw (f) node[blackvertex] {};
\draw (g) node[blackvertex] {};
\draw (h) node[blackvertex] {};

\draw (z) node[blackvertex] {};
\draw (y) node[blackvertex] {};

\draw (x) node[blackvertex] {};
\draw (w) node[blackvertex] {};
\draw (v) node[blackvertex] {};

\draw (1) node[vertex] {$\{\{\cdot|0\}|0\}$};
\draw (2) node[vertex] {$\{0|\{\cdot|0\}\}$};
\draw (3) node[vertex] {$1 = \{0|\cdot\}$};
\draw (4) node[vertex] {$* = \{0|0\}$};

\end{tikzpicture}
\end{center}
\caption{Some game trees}\label{fig:trees}
\end{figure}

Under both conventions, we can sort all games into four sets according to their outcomes.
When Left has a winning strategy on a game $G$ no matter which player starts, we say $G$ has outcome $\L$, and $G$ is an $\L$-position.
Similarly, $\N$, $\P$ and $\R$ (for Next, Previous and Right) denote respectively the outcomes of games on which the first player, the second player and Right has a winning strategy whoever starts the game.
The mis\`ere outcome of a game $G$ is denoted $o^-(G)$.
$\P$-positions are games in which players would rather have their opponent starts, that they would like to pass if it was their turn.
Outcomes are partially ordered according to Figure~\ref{fig:order}, with Left prefering greater games.

\delete{
\begin{figure}
\begin{center}

\begin{tikzpicture}
[thick,scale=1,
     vertex/.style={circle,draw=white!100,inner sep=1pt,minimum
size=2mm,fill=white!100},
     blackvertex/.style={circle,draw,inner sep=0pt,minimum
size=1.5mm,fill=black!100},
     clause/.style={circle,draw,inner sep=0pt,minimum
size=3mm,fill=white!100}]


\coordinate (L) at (0,2.8);
\coordinate (N) at (-1.4,1.4);
\coordinate (P) at (1.4,1.4);
\coordinate (R) at (0,0);

\draw (L)--(N)--(R)--(P)--(L);

\draw (L) node[vertex] {$\L$};
\draw (N) node[vertex] {$\N$};
\draw (R) node[vertex] {$\R$};
\draw (P) node[vertex] {$\P$};

\end{tikzpicture}
\vspace{-0.7cm}
\end{center}
\caption{Partial ordering of outcomes}\label{fig:order}
\end{figure}
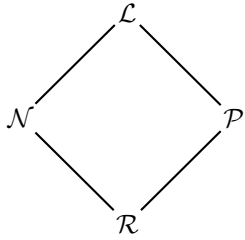
}

Given two games $G$ and $H$, we say that $G$ is greater than or equal to $H$ in mis\`ere play whenever Left always prefers the game $G$ rather than the game $H$, that is $G\geqslant^- H$ if for every game $X$, $o^-(G+X) \geqslant o^-(H+X)$.
We say that $G$ and $H$ are equivalent in misère play, denoted $G\equiv^- H$, when we have both $G \geqslant^- H$ and $H \geqslant^- G$.

General equivalence and comparison are very limited in general mis\`ere play (see \cite{partizan,canonical}), this is why Plambeck and Siegel defined in \cite{quotient1,quotient2} an equivalence relationship under restricted sets of games, leading to a breakthrough in the study of mis\`ere play games.

\begin{dfn}[\cite{quotient1,quotient2}]
Let $\UU$ be a set of games, $G$ and $H$ two games. 
We say $G$ is greater than or equal to $H$ modulo $\UU$ in mis\`ere play and write 
$G \geqslant^-_\UU H$ 
if $o^-(G+X) \geqslant o^-(H+X)$ for every $X \in \UU$. 
We say $G$ is equivalent to $H$ modulo $\UU$ in mis\`ere play and write $G \equiv^-_\UU H$ if $G \geqslant^-_\UU H$ and $H \geqslant^-_\UU G$.
\end{dfn}

\begin{figure}
\begin{center}

\begin{tikzpicture}
[thick,scale=1,
     vertex/.style={circle,draw=white!100,inner sep=1pt,minimum
size=2mm,fill=white!100},
     blackvertex/.style={circle,draw,inner sep=0pt,minimum
size=1.5mm,fill=black!100},
     clause/.style={circle,draw,inner sep=0pt,minimum
size=3mm,fill=white!100}]


\coordinate (L) at (0,2.8);
\coordinate (N) at (-1.4,1.4);
\coordinate (P) at (1.4,1.4);
\coordinate (R) at (0,0);

\draw (L)--(N)--(R)--(P)--(L);

\draw (L) node[vertex] {$\L$};
\draw (N) node[vertex] {$\N$};
\draw (R) node[vertex] {$\R$};
\draw (P) node[vertex] {$\P$};

\end{tikzpicture}
\vspace{-0.7cm}
\end{center}
\caption{Partial ordering of outcomes}\label{fig:order}
\end{figure}

Whenever $\UU$ is closed under sum, $\equiv^-_\UU$ is a congruence relation between elements of $\UU$.
Thus the disjunctive sum modulo $\UU$ defines a monoid $\MM_\UU = \UU/\equiv^-_\UU$.
We also consider the tetrapartition of $\MM_\UU$ according to outcomes: given an outcome $\OO$, we note $\OO_\UU$ the set of equivalence classes of $\UU$ with outcome $\OO$.
The structure $\QQ_\UU = (\MM_\UU, \P_\UU, \L_\UU, \R_\UU)$ is the {\em mis\`ere quotient} of $\UU$, as defined by Plambeck and Siegel in~\cite{quotient1,quotient2}, with the addition of $\L$ and $\R$ outcomes since we consider partizan games.

This approach gave quite some results.
For instance, Plambeck and Siegel~\cite{quotient1,quotient2} considered and solved the sets of all positions of given games, octal games in particular.
Other sets have been considered, including the sets of alternating games $\mathcal{A}$~\cite{alternating}, impartial games $\mathcal{I}$~\cite{ww,onag}, dicot games $\mathcal{D}$~\cite{mpg,dicotaxo,sprigs}, dead-ending games $\mathcal{E}$~\cite{pkayles,deadending}, and all games $\mathcal{G}$~\cite{partizan,canonical}.

We believe that some properties, namely being closed under disjunctive sum, conjugation and taking option, make a set more relevant to be studied.
We hence define a universe to be a set closed under disjunctive sum, conjugation and taking option.

Another interesting property for a game is to be dead-ending.
We say a Left (resp. Right) end is a dead end if all its followers are Left (resp. Right) ends.
A game is said to be dead-ending if all its end followers are dead ends.

In Section~\ref{sec:de}, we look at universes with no $\P$-positions, establish the invertibility of all elements when they are all dead-ending, and give an example of a universe with almost no invertible element when this last condition is dropped.
In Section~\ref{sec:z}, we focus on a particular quotient, $\qz$, and prove that if several universes share this quotient, then their sum shares this quotient as well.

\section{Invertibility modulo universes without $\P$-position}
\label{sec:de}

This section is dedicated to universes with no $\P$-position.
We first consider dead-ending games.

We recall the following lemma from~\cite{deadending}, that we use to prove invertibility of games.

\begin{lem}[\cite{deadending}]
\label{lem:inverse}
Let $\UU$ be a set of games closed under conjugation and taking option, and $S$ a set of games closed under taking option.
If \mbox{$G + \overline{G} + X \in \L^- \cup \N^-$} for every game $G \in S$ and every Left end $X \in \UU$, then $G + \overline{G} \equiv^-_\UU 0$ for every $G \in S$.
\end{lem}

We can now prove that all games are invertible in dead-ending universes containing no mis\`ere $\P$-position.

\begin{theo}
Let $\UU$ be a set of games closed under conjugation, sum, and taking option, such that every game in $\UU$ is dead-ending and no game in $\UU$ has mis\`ere outcome $\P$.
For any game $G$ in $\UU$, we have $G + \overline{G} \equiv^-_\UU 0$.
\end{theo}

\begin{proof}
By Lemma~\ref{lem:inverse}, we just need to prove that Left wins $G + \overline{G} + X$ playing first for every Left end $X \in \UU$ and every $G \in \UU$.
We actually prove that if $X = 0$, then $o^-(G + \overline{G} + X) = \N$, and otherwise $o^-(G + \overline{G} + X) = \L$ by induction on $G$ and $X$.
If $X=0$, then as $G + \overline{G} + X = \overline{G + \overline{G} + X}$, the outcome of the game is $\N$ or $\P$, but as no game in $\UU$ has outcome $\P$, its outcome is $\N$.

Assume now $X \neq 0$ and consider Right starts the game.
If he plays in $X$, he ends up in a position with outcome $\N$ or $\L$ by induction, where Left wins playing first.
If he plays in $G + \overline{G}$, Left can answer with the symmetric move and leave her opponent a position $G' + \overline{G'} + X$, with $G'$ an option of $G$, which has outcome $\L$ by induction.
Hence Left wins $G + \overline{G} + X$ playing second.
As no game in $\UU$ has outcome $\P$, we have $o^-(G + \overline{G} + X) = \L$.

We thus have the hypothesis of Lemma~\ref{lem:inverse} and can conclude that every game of $\UU$ is invertible with its conjugate as inverse.
\end{proof}

Unfortunately, that property is not true for all universes, as we now give a counterexample in the general case.
We define a game $a = \{\cdot|2\}$, and we look at the closure $\cla$ of $a$ and its conjugate under sum and followers.
Since $1+1=2$, an element of $\cla$ can be written under the form $k_1 a + k_2 \overline{a} + k_3 1 + k_4 \overline{1}$ (see Figure~\ref{fig:a1}).
Note that neither $a$ nor $\overline a$ is dead-ending.

\begin{figure}
\begin{center}

\begin{tikzpicture}
[thick,scale=1,
     vertex/.style={circle,draw=white!100,inner sep=1pt,minimum
size=2mm,fill=white!100},
     blackvertex/.style={circle,draw,inner sep=0pt,minimum
size=1.5mm,fill=black!100},
     clause/.style={circle,draw,inner sep=0pt,minimum
size=3mm,fill=white!100}]


\coordinate (1) at (-1.75,1);
\coordinate (2) at (-2,0.5);

\coordinate (a1) at (0,1.5);
\coordinate (a2) at (0.5,1);
\coordinate (a3) at (0.25,0.5);
\coordinate (a4) at (0,0);

\coordinate (b1) at (3,1.5);
\coordinate (b2) at (2.5,1);
\coordinate (b3) at (2.75,0.5);
\coordinate (b4) at (3,0);

\coordinate (3) at (4.75,1);
\coordinate (4) at (5,0.5);

\coordinate (n) at (-1.875,-0.5);
\coordinate (a) at (0.25,-0.5);
\coordinate (b) at (2.75,-0.5);
\coordinate (m) at (4.875,-0.5);

\draw (a1)--(a2)--(a3)--(a4);
\draw (b1)--(b2)--(b3)--(b4);
\draw (1)--(2);
\draw (3)--(4);


\draw (1) node[blackvertex] {};
\draw (2) node[blackvertex] {};

\draw (a1) node[blackvertex] {};
\draw (a2) node[blackvertex] {};
\draw (a3) node[blackvertex] {};
\draw (a4) node[blackvertex] {};

\draw (b1) node[blackvertex] {};
\draw (b2) node[blackvertex] {};
\draw (b3) node[blackvertex] {};
\draw (b4) node[blackvertex] {};

\draw (3) node[blackvertex] {};
\draw (4) node[blackvertex] {};

\draw (n) node[vertex] {$1$};
\draw (a) node[vertex] {$a$};
\draw (b) node[vertex] {$\overline{a}$};
\draw (m) node[vertex] {$\overline{1}$};

\end{tikzpicture}
\vspace{-0.7cm}
\end{center}
\caption{The four generators of $\cla$}\label{fig:a1}
\end{figure}

We first fully determine the outcomes of games in $\cla$.

\begin{theo}
\label{th:aout}
Let $G$ be a game in $\cla$ and write $G = k_1 a + k_2 \overline{a} + k_3 1 + k_4 \overline{1}$.
We have
$$
o^-(G)=
\begin{cases}
\N & \text{if } k_1+k_3 = k_2+k_4 \text{ or } (k_2=k_3=0 \text{ and } k_1 \geqslant k_4) \text{ or } (k_1=k_4=0 \text{ and } k_2 \geqslant k_3) \\
\L & \text{if } k_1+k_4 > 0 \text{ and } k_2+k_4 > k_1+k_3 \\
\R & \text{if } k_2+k_3 > 0 \text{ and } k_1+k_3 > k_2+k_4
\end{cases}
$$
\end{theo}

\begin{proof}
We prove the result by induction on $G$.
If $k_2=k_3=0$, then $G$ is a Left end and $o^-(G) \geqslant \N$.
Similarly, if $k_1=k_4=0$, then $G$ is a Right end and $o^-(G) \leqslant \N$.

Assume first $k_2=k_3=0$ and $k_1 \geqslant k_4$.
If Right moves first, either there is no move and he wins immediately, or he can play in one of the $k_1$ $a$, moving from $G$ to $(k_1-1)a + 2 \cdot 1 + k_4 \overline{1}$ which has outcome $\R$ by induction since $k_1-1+2 > k_1 \geqslant k_4$.
Similarly, if $k_1=k_4=0$ and $k_2 \geqslant k_3$, Left wins playing first.

Assume now both $k_1+k_4$ and $k_2+k_3$ are positive.
Playing first, Left can either play on an $\overline{a}$ or a $1$, both increasing the difference $(k_2+k_4)-(k_1+k_3)$ by $1$ while not changing the fact that $k_1+k_4$ is positive.
According to the induction, if that difference was non-negative, she moved to an $\L$-position, and otherwise, she moved either to an $\N$-position or an $\R$-position, both of which she loses playing second.
The symmetric result when Right plays first concludes the proof.
\end{proof}

Note that no game has outcome $\P$.
Using this characterisation, we can now prove that there are games in $\cla$ that are not invertible.
Worse, actually, only one of them is.

\begin{prop}
Let $G$ be a game in $\cla$.
If $G \equiv^-_{\cla} 0$, then $G = 0$.
\end{prop}

\begin{proof}
We write $G = k_1 a + k_2 \overline{a} + k_3 1 + k_4 \overline{1}$.
Assume $k_1+k_2+k_3+k_4$ is positive.
Then at least one player has a move.
Without loss of generality, we can assume Left has a move.
Now consider the game $X=(k_1+k_2+k_3+k_4+1)a$.
By Theorem~\ref{th:aout}, we have $o^-(0+X) = \N$, while $o^-(G+X) = \R$.
Hence $G \not\equiv^-_{\cla} 0$.
\end{proof}

Actually, the situation is even worse.
These games are not even cancellative.

\begin{prop}
The only cancellative game in $\cla$ is $0$.
\end{prop}

\begin{proof}
First note that $1+1+\overline{1}$ and $1$ are not equivalent modulo $\cla$ since $o^-(1+1+\overline{1}+\overline{a}+\overline{a}) = \L$ while $o^-(1+\overline{a}+\overline{a}) = \N$.
Similarly, $1+\overline{1}+\overline{1}$ and $\overline{1}$ are not equivalent modulo $\cla$.

Now let $G$ be any non-zero game in $\cla$.
$G$ has either a Left or a Right option (or both).
Without loss of generality, we may assume it has a Right option.
We claim that $G + 1$ and $G + 1+1+\overline{1}$ are equivalent modulo $\cla$.
Indeed, consider any game $X$ in $\cla$ and write $G+X+1 = k_1 a + k_2 \overline{a} + k_3 1 + k_4 \overline{1}$.
As $G$ is not a Right end, we have $k_1+k_4 > 0$.
We ensured $k_2+k_3 > 0$ by having $1$ in the sum.
Hence the outcome of $G+X+1$ is fully determined by $(k_1+k_3) - (k_2+k_4)$.
Similarly, the outcome of $G+X+1+1+\overline{1}$ is fully determined by $(k_1+(k_3+1) - (k_2+(k_4+1))$.
Since the two numbers are equal, the two games have the same outcome.
Hence we have $G + 1 \equiv^-_{\cla} G + 1+1+\overline{1}$ but $1 \not\equiv^-_{\cla} 1+1+\overline{1}$, thus completing the proof that $G$ is not cancellative in $\cla$.
\end{proof}

Nevertheless, there exist universes with non-dead-ending positions but without $\P$-positions where all games are invertible.
It could thus be interesting to characterise which ones among them share this property.

\section{The quotient $\qz$}
\label{sec:z}

The simplest non-trivial universe with no $\P$-position is $\mz = \{k_1 1 + k_2 \overline{1}|k_1,k_2 \in \NN\}$.
As \mbox{$1+\overline{1} \equiv^-_\mz 0$}~(\cite{deadending}), each equivalence class has a representative of which at most one of $k_1$ and $k_2$ is positive.
We note $\qz = (\langle 1,\overline 1| 1 + \overline 1 = 0\rangle,\emptyset,(\NN^*) \overline 1, (\NN^*) 1)$ the quotient of $\mz$.
Actually, several other universes~\cite{deadending, wothese}, seemingly more complex, share this same quotient.

\begin{prop}
\label{prop:Z}
A universe $\UU$ has quotient $\qz$ if there exists a surjective function $f: \UU \rightarrow \ZZ$ such that:
\begin{enumerate}[(i)]
\item $\forall G,H \in \UU, f(G+H) = f(G) + f(H)$,
\item $o^-(G) = \left\{ \begin{array}{ll}
\N & \text{ if } f(G) = 0, \\
\L & \text{ if } f(G) < 0, \\
\R & \text{ if } f(G) > 0.
\end{array}\right.$
\end{enumerate}
\end{prop}

In this case, we say that $f$ is a witness function of $\UU$.
We will see after Lemma~\ref{lem:Z=>n+1} that the witness function is actually unique.

Note that universes with positions that are not dead-ending can still have quotient $\qz$.
Still, all the positions are invertible.

We also define the sum of two sets of games as follows:

\begin{dfn}
Let $S_1$ and $S_2$ be two sets.
We define $S_1+S_2$ the sum of these sets as follows:
$$S_1+S_2 = \{s_1+s_2|s_1 \in S_1 \wedge s_2 \in S_2\}.$$
\end{dfn}

In~\cite{wothese}, the author considered sums of universes having quotient $\qz$ and these sums were sharing the same quotient $\qz$.
We here prove that this is always the case.

First, we give another characterisation for a universe to have quotient $\qz$.
The next lemma shows one way of the equivalence.


\begin{lem}
\label{lem:=>Z}
Let $\UU$ be a universe and $f: \UU \rightarrow \ZZ$ a surjective function such that:
\begin{enumerate}[a)]
\item $\forall G \in \UU, n>0, (f(G) = n) \Rightarrow ((\exists G^L \in \GL, f(G^L) = n-1) \wedge (\forall G^L \in \GL, f(G^L) \geqslant n-1))$,
\item $\forall G \in \UU, n \geqslant 0, (f(G) = n) \Rightarrow
\left\{ \begin{array}{r@{\,}l}
( & ((\GR \neq \emptyset) \Rightarrow (\exists G^R \in \GR, 1 \leqslant f(G^R) \leqslant n+1)) \\
\wedge & (\forall G^R \in \GR, f(G^R) \leqslant n+1))
\end{array}\right.,$
\item $\forall G \in \UU, n<0, (f(G) = n) \Rightarrow ((\exists G^R \in \GR, f(G^R) = n+1) \wedge (\forall G^R \in \GR, f(G^R) \leqslant n+1))$,
\item $\forall G \in \UU, n \leqslant 0, (f(G) = n) \Rightarrow
\left\{ \begin{array}{r@{\,}l}
( & ((\GL \neq \emptyset) \Rightarrow (\exists G^L \in \GL, -1 \geqslant f(G^L) \geqslant n-1)) \\
\wedge & (\forall G^L \in \GL, f(G^L) \geqslant n-1))
\end{array}\right..$
\end{enumerate}
Then $\UU$ has quotient $\qz$, having $f$ as a witness function.
\end{lem}

\begin{proof}
We prove that $f$ satisfies the two conditions of Proposition~\ref{prop:Z} by induction on the games in $\UU$.
First consider a game $G$ in $\UU$.
If $G$ has no option, then it cannot satisfy the right part of the implications $a)$ and $c)$.
Hence $f(G) = 0$ which corresponds since $o^-(G) = \N$.

Assume $G$ is a game such that $f(G) > 0$.
From $a)$, it has a Left option, and all its Left options have mis\`ere outcome $\N$ or $\R$ by induction.
Hence Right has a winning strategy playing second in $G$.
From $b)$, Right can move to a mis\`ere $\R$-position by induction if he has any move.
Hence Right has a winning strategy playing first in $G$, which proves $G$ has mis\`ere outcome $\R$.

We can prove similarly that if $f(G) < 0$, $G$ has mis\`ere outcome $\L$.

Now assume $f(G) = 0$.
From $b)$, Right can move to a mis\`ere $\R$-positions by induction if he has any move.
Hence Right has a winning strategy playing first in $G$.
Similarly, from $d)$, Left has a winning strategy playing first in $G$.
Hence $G$ has mis\`ere outcome $\N$.

Now consider two games $G$ and $H$ in $\UU$.
Assume first $f(G) + f(H) > 0$.
Then at least one among $f(G)$ and $f(H)$ is positive.
Without loss of generality, we may assume $f(G)$ is positive.
By $a)$, there exists a Left option $G^{L_1}$ of $G$ such that $f(G^{L_1}) = f(G)-1$.
Hence we have a Left option $G^{L_1} + H$ of $G+H$ such that $f(G^{L_1} + H) = f(G^{L_1}) + f(H) = f(G) - 1 + f(H)$ by induction.
By $a)$, we have $f(G+H)$ is at most $f(G) + f(H)$.
Similarly, as all Left options of $G+H$ are of the form $G^L + H$ or $G + H^L$, using $a)$ and $d)$ and induction, we can say that they are all mapped to integers greater than or equal to $f(G)+f(H)-1$.
By $d)$, as $G+H$ has a Left option and all Left options mapped to non-negative integers, $f(G+H)$ is positive.
For any positive integer $k$ less than $f(G)+f(H)$, there is no Left move from $G+H$ to a position mapped to $k-1$, so by $a)$, $f(G+H)$ cannot be $k$.
Hence $f(G+H) = f(G) + f(H)$.

Similarly, we have that if $f(G) + f(H) < 0$, then $f(G+H) = f(G) + f(H)$.

Now assume $f(G) + f(H) = 0$.
First assume $f(G) = f(H) = 0$.
If $G+H$ has no Left option, it cannot satisfy the right part of the implication $a)$, hence it is mapped to a non-positive integer.
Assume now $G+H$ has a Left option, it means $G$ or $H$ has a Left option.
Without loss of generality, we may assume $G$ has a Left option.
From $d)$, we know there exists a Left option $G^{L_1}$ of $G$ such that $f(G^{L_1}) = -1$.
Hence the Left option $G^{L_1} + H$ of $G+H$ is such that $f(G^{L_1} + H) = f(G^L_1) + f(H) = -1$ by induction, which implies that $f(G+H)$ is non-positive, as the opposite would contradict $a)$.
Similarly, we prove that $f(G+H)$ is non-negative.
Hence we have $f(G+H) = 0$.
Assume now without loss of generality that $f(G) > 0 > f(H)$.
From $a)$, we get a Left option $G^{L_1}$ of $G$ such that $f(G^{L_1}) = f(G)-1$.
Similarly as above, this implies $f(G+H)$ is non-positive.
We prove that $f(G+H)$ is non-negative in a similar way.
Hence we have $f(G+H) = 0$, which concludes the proof.
\end{proof}

We now prove the other way in four steps.
First we show that players cannot have too good moves.

\begin{lem}
\label{lem:Z=>n-1}
Let $\UU$ be a universe with quotient $\qz$ and $f$ a witness function of $\UU$.
Then for any game $G$ in $\UU$, for any Left option $G^L$ of $G$, we have $f(G^L) \geqslant f(G) - 1$.
\end{lem}

\begin{proof}
Let $G$ be a game in $\UU$ and $n$ the image of $G$ through $f$.
As $f$ is surjective, we can find a game $H$ in $\UU$ such that $f(H) = 1-n$.
We have $f(G+H) = 1$, so $G+H$ has mis\`ere outcome $\R$.
Hence any first move of Left is losing, which means their images through $f$ should be at least $n-1$.
\end{proof}

We now show that when Left is supposed to win playing first and has a move, she has an interesting move.

\begin{lem}
\label{lem:Z=>exn-1}
Let $\UU$ be a universe with quotient $\qz$ and $f$ a witness function of $\UU$.
Let $G$ be a game in $\UU$ such that $f(G)$ is non-positive and $G$ has a Left option.
Then $G$ has a Left option $G^L$ such that $-1 \geqslant f(G^L) \geqslant f(G)-1$.
\end{lem}

\begin{proof}
As $f(G)$ is non-positive, $G$ has mis\`ere outcome $\N$ or $\L$.
Hence Left wins playing first in $G$, either because $G$ has no Left option or because she has a winning move.
$G$ having a Left option puts us in the second case and there is a winning Left move from $G$ to some $G^L$.
By Lemma~\ref{lem:Z=>n-1}, we have $f(G^L) \geqslant f(G)-1$.
Since $G^L$ is a winning move, it has outcome $\P$ or $\L$.
Hence $f(G^L)$ is negative.
Therefore, we have $-1 \geqslant f(G^L) \geqslant f(G)-1$.
\end{proof}

For the next part, we need to ensure we can find Right ends in $\UU$ whose images through the witness function cover all positive number.
Hence we consider the game $\{0|\cdot\}$.

\begin{lem}
\label{lem:Z=>0.1}
Let $\UU$ be a universe with quotient $\qz$ and $f$ a witness function of $\UU$.
Then $\{0|\cdot\} \in \UU$ and $f(\{0|\cdot\})=1$.
\end{lem}

\begin{proof}
As $f$ is surjective, there is an infinite number of games in $\UU$.
As $\UU$ is closed under taking options, there exists some game in $\UU$ with birthday $1$.
The three games with birthday $1$ are $\{0|\cdot\}$, $\{0|0\}$ and $\{\cdot|0\}$.
As $\{0|0\}$ has mis\`ere outcome $\P$, it cannot be in $\UU$.
As $\UU$ is closed under conjugation and $\{0|\cdot\}$ and $\{\cdot|0\}$ are each other's conjugates, both are in $\UU$.
$\{0|\cdot\}$ has mis\`ere outcome $\R$, hence its image through $f$ must be positive.
Similarly, $0$ has mis\`ere outcome $\N$ and its image through $f$ is $0$.
Having a Left option to $0$, whose image through $f$ is $0$, $\{0|\cdot\}$'s image through $f$ must be at most $1$ by Lemma~\ref{lem:Z=>n-1}.
Hence $f(\{0|\cdot\})=1$.
\end{proof}

We can now prove that when Right loses playing first, he has a move that is not too bad.

\begin{lem}
\label{lem:Z=>n+1}
Let $\UU$ be a universe with quotient $\qz$ and $f$ a witness function of $\UU$.
Let $G$ be a game in $\UU$ such that $f(G)$ is negative.
Then $G$ has a Right option $G^R$ such that $f(G^R) = f(G)+1$.
\end{lem}

\begin{proof}
As $f(G)$ is negative, $G$ has mis\`ere outcome $\L$, and Right has a move in $G$.

By Lemma~\ref{lem:Z=>0.1}, we have $\{0|\cdot\}$ in $\UU$ and $f(\{0|\cdot\}) = 1$.
As $\UU$ is closed under sum, we have \mbox{$(-f(G)) \cdot \{0|\cdot\}$} in $\UU$.
We have $f(G + ((-f(G)) \cdot \{0|\cdot\})) = 0$ so $G + ((-f(G)) \cdot \{0|\cdot\})$ has mis\`ere outcome $\N$.
Hence Right has a winning move in $G + ((-f(G)) \cdot \{0|\cdot\})$, which has to be in $G$ since he has no move in $(-f(G)) \cdot \{0|\cdot\}$.
Therefore, there is a Right move from $G$ to some $G^R$ such that $f(G^R) > f(G)$.
By the conjugate version of Lemma~\ref{lem:Z=>n-1}, no Right option of $G$ may have an image through $f$ more than $f(G)+1$.
Hence $G^R$ is such that $f(G^R) = f(G)+1$.
\end{proof}

Note that this proof implies that the witness function of a universe with quotient $\qz$ is unique.
If $f'$ is a witness function of $\UU$, then $f'(\{0|\cdot\}) = 1$ and for any $G$ in $\UU$, if $f(G)$ is negative, $G + ((-f(G)) \cdot \{0|\cdot\})$ has outcome $\N$ so $f'(G + ((-f(G)) \cdot \{0|\cdot\}))=0$ and $f'(G) = f(G)$; similarly when $f(G)$ is positive; when $f(G)$ is $0$, $G$ has outcome $\N$ so $f'(G)$ is $0$ as well.

We can now state the other way of the characterisation.

\begin{theo}
\label{theo:Z<=>}
Let $\UU$ be a universe with quotient $\qz$.
Then the witness function $f$ of $\UU$ is such that:
\begin{enumerate}[a)]
\item $\forall G \in \UU, n>0, (f(G) = n) \Rightarrow ((\exists G^L \in \GL, f(G^L) = n-1) \wedge (\forall G^L \in \GL, f(G^L) \geqslant n-1))$,
\item $\forall G \in \UU, n \geqslant 0, (f(G) = n) \Rightarrow
\left\{ \begin{array}{r@{\,}l}
( & ((\GR \neq \emptyset) \Rightarrow (\exists G^R \in \GR, 1 \leqslant f(G^R) \leqslant n+1)) \\
\wedge & (\forall G^R \in \GR, f(G^R) \leqslant n+1))
\end{array}\right.,$
\item $\forall G \in \UU, n<0, (f(G) = n) \Rightarrow ((\exists G^R \in \GR, f(G^R) = n+1) \wedge (\forall G^R \in \GR, f(G^R) \leqslant n+1))$,
\item $\forall G \in \UU, n \leqslant 0, (f(G) = n) \Rightarrow
\left\{ \begin{array}{r@{\,}l}
( & ((\GL \neq \emptyset) \Rightarrow (\exists G^L \in \GL, -1 \geqslant f(G^L) \geqslant n-1)) \\
\wedge & (\forall G^L \in \GL, f(G^L) \geqslant n-1))
\end{array}\right..$
\end{enumerate}
\end{theo}

\begin{proof}
The result is a combination of Lemmas~\ref{lem:Z=>n-1}, \ref{lem:Z=>exn-1}, \ref{lem:Z=>n+1} and their conjugate versions.
\end{proof}

With this characterisation, we can now prove the main theorem of this section.

\begin{theo}
Let $\UU_1$ and $\UU_2$ be two universes with quotient $\qz$.
Then $\Usum$ is a universe having quotient $\qz$.
\end{theo}

\begin{proof}
Let $G$ and $H$ be two games in $\Usum$.
We can write $G = G_1 + G_2$ and $H = H_1 + H_2$ such that $G_1,H_1 \in \UU_1$ and $G_2,H_2 \in \UU_2$.
Then $\overline{G} = \overline{G_1} + \overline{G_2} \in \Usum$. 
Hence $\Usum$ is closed under conjugation.
$G+H = (G_1+H_1) + (G_2+H_2) \in \Usum$. 
Hence $\Usum$ is closed under disjunctive sum.
An option of $G$ is either of the form $G_1$ summed with an option of $G_2$ or of the form $G_2$ summed with an option of $G_1$.
Hence $\Usum$ is closed under taking options.
Therefore $\Usum$ is a universe.

Let $f_1$ and $f_2$ be the witness functions of $\UU_1$ and $\UU_2$ respectively.
We define $f:\Usum\rightarrow\ZZ$ as $f(G_1+G_2) = f_1(G_1)+f_2(G_2)$ for any $G_1$ in $\UU_1$ and $G_2$ in $\UU_2$.
Let $G = G_1 + G_2$ be a game in $\Usum$ such that $G_1 \in \UU_1$ and $G_2 \in \UU_2$.
We prove by induction on $G$ that $f$ satisfies the hypothesis of Lemma~\ref{lem:=>Z}, using the fact that both $f_1$ and $f_2$ satisfy these hypothesis by Theorem~\ref{theo:Z<=>}.
Any Left option of $G$ is of the form $G_1+G_2^L$ or $G_1^L+G_2$.
As any $G_1^L$ and $G_2^L$ are such that $f_1(G_1^L) \geqslant f_1(G_1) - 1$ and $f_2(G_2^L) \geqslant f_2(G_2) - 1$, we have $f(G^L) \geqslant f(G) - 1$ for any Left option $G^L$ of $G$.
Similarly, we have $f(G^R) \leqslant f(G) + 1$ for any Right option $G^R$ of $G$.

Assume first $f(G)$ is positive.
Then $f_1(G_1)$ or $f_2(G_2)$ is positive.
Without loss of generality, we may assume $f_1(G_1)$ to be positive.
Then there exists a Left option $G_1^L$ of $G_1$ such that $f_1(G_1^L) = f_1(G_1) - 1$.
Hence the Left option $G_1^L + G_2$ of $G_1 + G_2$ has an image through $f$ with value $f(G)-1$.
Similarly, if $f(G)$ is negative, there exists a Right option $G^R$ of $G$ such that $f(G^R) = f(G)+1$.

Assume now $f(G)$ is non-negative.
If there is no Right option from $G$, it is fine.
Assume then Right has a move from $G$.
If $f_1(G_1)$ is negative, there is a Right option $G_1^R$ of $G_1$ such that $f_1(G_1^R) = f_1(G_1) + 1$, hence the Right option $G_1^R + G_2$ of $G_1+G_2$ has an image through $f$ with value $f(G)+1$.
Similarly, there is a Right option $G_1 + G_2^R$ of $G_1+G_2$ having an image through $f$ with value $f(G)+1$ whenever $f_2(G_2)$ is negative.
Assume both $f_1(G_1)$ and $f_2(G_2)$ are non-negative.
Without loss of generality, since Right has an option from $G_1+G_2$, we may assume $G_1$ has a Right option.
As $f_1(G_1)$ is non-negative, there exists a Right option $G_1^R$ of $G_1$ such that $1 \leqslant f_1(G_1^R) \leqslant f_1(G_1) + 1$.
Then the Right option $G_1^R + G_2$ of $G_1 + G_2$ is such that $1 \leqslant 1 + f_2(G_2) \leqslant f(G_1^R + G_2) \leqslant f(G) + 1$.
Similarly, if $f(G)$ is non-positive, either Left has no move from $G$ or there exists a Left option $G^L$ of $G$ such that $-1 \geqslant f(G^L) \geqslant f(G) - 1$.

Therefore $\Usum$ is a universe with quotient $\qz$.
\end{proof}

This result does not seem to generalise easily to other quotients since we had to look at the possible moves of all positions in every equivalence class.
Actually, the result is not true for any quotient.
Call $*$ the game $\{0|0\}$, $*2$ the game $\{0,*|0,*\}$, $2_\#$ the game $\{*2|*2\}$, $2_\#0$ the game $\{0,2_\#|0,2_\#\}$ and $2_\#20$ the game $\{0,*2,2_\#|0,*2,2_\#\}$.
Plambeck~\cite{eqctrex} found that the closures $\cl(2_\#0)$ and $\cl(2_\#20)$ by sum and followers of the last two games share the same quotient, but their sum does not.
Their common quotient can be seen as the following:
$$(\langle a,b,c|a^2=1,b^3=b,b^2c=c,c^3=ac^2\rangle,\{a,b^2,bc,c^2\},\emptyset,\emptyset)$$ with fourteen elements.
What is surprising, and might explain why the sum gets a bigger quotient, is that the common elements of these two universes are not always mapped to the same element of the quotient.
For example, $*2$ is mapped to $b$ from $\cl(2_\#0)$ but to $ab$ from $\cl(2_\#20)$.
This situation cannot happen with $\qz$ because of Lemma~\ref{lem:Z=>0.1}.

\delete{
\begin{quest}
Do we have that whenever two universes $\UU_1$ and $\UU_2$ share the same quotient, then their sum $\Usum$ shares this quotient as well?
\end{quest}

Notice that we consider the tetrapartition of the elements as a part of the quotient, otherwise we would already know of a counterexample.
Call $*$ the game $\{0|0\}$, $*2$ the game $\{0,*|0,*\}$, $2_\#$ the game $\{*2|*2\}$, $2_\#0$ the game $\{0,2_\#|0,2_\#\}$ and $2_\#20$ the game $\{0,*2,2_\#|0,*2,2_\#\}$ (see Figure~\ref{fig:ctrex}).
Plambeck~\cite{eqctrex} found that the quotients of the closures $\cl(2_\#0)$ and $\cl(2_\#20)$ by sum and followers of the last two games share the same fourteen elements, but their sum has a quotient with thirty elements.
However, these two quotients do not share their $\P$-positions, hence they are not completely equal and it is not so much of a surprise the quotient of the sum gets even more different.
}

\begin{figure}
\begin{center}

\begin{tikzpicture}
[thick,scale=1,
     vertex/.style={circle,draw=white!100,inner sep=1pt,minimum
size=2mm,fill=white!100},
     blackvertex/.style={circle,draw,inner sep=0pt,minimum
size=1.5mm,fill=black!100},
     clause/.style={circle,draw,inner sep=0pt,minimum
size=3mm,fill=white!100}]


\coordinate (1) at (-6,7);
\coordinate (2) at (-5.5,8);
\coordinate (3) at (-5,7);
\coordinate (4) at (-5,8);
\coordinate (5) at (-4.75,9);
\coordinate (6) at (-4.5,8);
\coordinate (7) at (-4.5,7);
\coordinate (8) at (-4,8);
\coordinate (9) at (-3.5,7);

\coordinate (10) at (-3.25,10);

\coordinate (11) at (-3,7);
\coordinate (12) at (-2.5,8);
\coordinate (13) at (-2,7);
\coordinate (14) at (-2,8);
\coordinate (15) at (-1.75,9);
\coordinate (16) at (-1.5,8);
\coordinate (17) at (-1.5,7);
\coordinate (18) at (-1,8);
\coordinate (19) at (-0.5,7);

\coordinate (20) at (-1,10);
\coordinate (21) at (0,11);
\coordinate (22) at (1,10);

\coordinate (23) at (0.5,7);
\coordinate (24) at (1,8);
\coordinate (25) at (1.5,7);
\coordinate (26) at (1.5,8);
\coordinate (27) at (1.75,9);
\coordinate (28) at (2,8);
\coordinate (29) at (2,7);
\coordinate (30) at (2.5,8);
\coordinate (31) at (3,7);

\coordinate (32) at (3.25,10);

\coordinate (33) at (3.5,7);
\coordinate (34) at (4,8);
\coordinate (35) at (4.5,7);
\coordinate (36) at (4.5,8);
\coordinate (37) at (4.75,9);
\coordinate (38) at (5,8);
\coordinate (39) at (5,7);
\coordinate (40) at (5.5,8);
\coordinate (41) at (6,7);

\coordinate (42) at (0,6);

\draw (1)--(2)--(3);
\draw (4)--(5)--(6);
\draw (7)--(8)--(9);
\draw (2)--(5)--(8);

\draw (5)--(10)--(15);

\draw (11)--(12)--(13);
\draw (14)--(15)--(16);
\draw (17)--(18)--(19);
\draw (12)--(15)--(18);

\draw (20)--(21)--(22);
\draw (10)--(21)--(32);

\draw (23)--(24)--(25);
\draw (26)--(27)--(28);
\draw (29)--(30)--(31);
\draw (24)--(27)--(30);

\draw (27)--(32)--(37);

\draw (33)--(34)--(35);
\draw (36)--(37)--(38);
\draw (39)--(40)--(41);
\draw (34)--(37)--(40);


\draw (1) node[blackvertex] {};
\draw (2) node[blackvertex] {};
\draw (3) node[blackvertex] {};
\draw (4) node[blackvertex] {};
\draw (5) node[blackvertex] {};
\draw (6) node[blackvertex] {};
\draw (7) node[blackvertex] {};
\draw (8) node[blackvertex] {};
\draw (9) node[blackvertex] {};
\draw (10) node[blackvertex] {};
\draw (11) node[blackvertex] {};
\draw (12) node[blackvertex] {};
\draw (13) node[blackvertex] {};
\draw (14) node[blackvertex] {};
\draw (15) node[blackvertex] {};
\draw (16) node[blackvertex] {};
\draw (17) node[blackvertex] {};
\draw (18) node[blackvertex] {};
\draw (19) node[blackvertex] {};
\draw (20) node[blackvertex] {};
\draw (21) node[blackvertex] {};
\draw (22) node[blackvertex] {};
\draw (23) node[blackvertex] {};
\draw (24) node[blackvertex] {};
\draw (25) node[blackvertex] {};
\draw (26) node[blackvertex] {};
\draw (27) node[blackvertex] {};
\draw (28) node[blackvertex] {};
\draw (29) node[blackvertex] {};
\draw (30) node[blackvertex] {};
\draw (31) node[blackvertex] {};
\draw (32) node[blackvertex] {};
\draw (33) node[blackvertex] {};
\draw (34) node[blackvertex] {};
\draw (35) node[blackvertex] {};
\draw (36) node[blackvertex] {};
\draw (37) node[blackvertex] {};
\draw (38) node[blackvertex] {};
\draw (39) node[blackvertex] {};
\draw (40) node[blackvertex] {};
\draw (41) node[blackvertex] {};

\draw (42) node[vertex] {$2_\#0$};

\coordinate (01) at (-6,0);
\coordinate (02) at (-5.5,1);
\coordinate (03) at (-5,0);
\coordinate (04) at (-5,1);
\coordinate (05) at (-4.75,2);
\coordinate (06) at (-4.5,1);
\coordinate (07) at (-4.5,0);
\coordinate (08) at (-4,1);
\coordinate (09) at (-3.5,0);

\coordinate (010) at (-3.25,3);

\coordinate (011) at (-3,0);
\coordinate (012) at (-2.5,1);
\coordinate (013) at (-2,0);
\coordinate (014) at (-2,1);
\coordinate (015) at (-1.75,2);
\coordinate (016) at (-1.5,1);
\coordinate (017) at (-1.5,0);
\coordinate (018) at (-1,1);
\coordinate (019) at (-0.5,0);

\coordinate (020) at (-1,3);
\coordinate (021) at (0,4);
\coordinate (022) at (1,3);

\coordinate (023) at (0.5,0);
\coordinate (024) at (1,1);
\coordinate (025) at (1.5,0);
\coordinate (026) at (1.5,1);
\coordinate (027) at (1.75,2);
\coordinate (028) at (2,1);
\coordinate (029) at (2,0);
\coordinate (030) at (2.5,1);
\coordinate (031) at (3,0);

\coordinate (032) at (3.25,3);

\coordinate (033) at (3.5,0);
\coordinate (034) at (4,1);
\coordinate (035) at (4.5,0);
\coordinate (036) at (4.5,1);
\coordinate (037) at (4.75,2);
\coordinate (038) at (5,1);
\coordinate (039) at (5,0);
\coordinate (040) at (5.5,1);
\coordinate (041) at (6,0);

\coordinate (042) at (0,-1);

\coordinate (51) at (-6,1);
\coordinate (52) at (-6.25,2);
\coordinate (53) at (-6.5,1);
\coordinate (54) at (-6.5,2);
\coordinate (55) at (-6.75,3);
\coordinate (56) at (-7,2);
\coordinate (57) at (-7,1);
\coordinate (58) at (-7.25,2);
\coordinate (59) at (-7.5,1);

\coordinate (61) at (6,1);
\coordinate (62) at (6.25,2);
\coordinate (63) at (6.5,1);
\coordinate (64) at (6.5,2);
\coordinate (65) at (6.75,3);
\coordinate (66) at (7,2);
\coordinate (67) at (7,1);
\coordinate (68) at (7.25,2);
\coordinate (69) at (7.5,1);

\draw (01)--(02)--(03);
\draw (04)--(05)--(06);
\draw (07)--(08)--(09);
\draw (02)--(05)--(08);

\draw (05)--(010)--(015);

\draw (011)--(012)--(013);
\draw (014)--(015)--(016);
\draw (017)--(018)--(019);
\draw (012)--(015)--(018);

\draw (020)--(021)--(022);
\draw (010)--(021)--(032);

\draw (023)--(024)--(025);
\draw (026)--(027)--(028);
\draw (029)--(030)--(031);
\draw (024)--(027)--(030);

\draw (027)--(032)--(037);

\draw (033)--(034)--(035);
\draw (036)--(037)--(038);
\draw (039)--(040)--(041);
\draw (034)--(037)--(040);

\draw (51)--(52)--(53);
\draw (54)--(55)--(56);
\draw (57)--(58)--(59);
\draw (52)--(55)--(58);

\draw (61)--(62)--(63);
\draw (64)--(65)--(66);
\draw (67)--(68)--(69);
\draw (62)--(65)--(68);

\draw (55)--(021)--(65);


\draw (01) node[blackvertex] {};
\draw (02) node[blackvertex] {};
\draw (03) node[blackvertex] {};
\draw (04) node[blackvertex] {};
\draw (05) node[blackvertex] {};
\draw (06) node[blackvertex] {};
\draw (07) node[blackvertex] {};
\draw (08) node[blackvertex] {};
\draw (09) node[blackvertex] {};
\draw (010) node[blackvertex] {};
\draw (011) node[blackvertex] {};
\draw (012) node[blackvertex] {};
\draw (013) node[blackvertex] {};
\draw (014) node[blackvertex] {};
\draw (015) node[blackvertex] {};
\draw (016) node[blackvertex] {};
\draw (017) node[blackvertex] {};
\draw (018) node[blackvertex] {};
\draw (019) node[blackvertex] {};
\draw (020) node[blackvertex] {};
\draw (021) node[blackvertex] {};
\draw (022) node[blackvertex] {};
\draw (023) node[blackvertex] {};
\draw (024) node[blackvertex] {};
\draw (025) node[blackvertex] {};
\draw (026) node[blackvertex] {};
\draw (027) node[blackvertex] {};
\draw (028) node[blackvertex] {};
\draw (029) node[blackvertex] {};
\draw (030) node[blackvertex] {};
\draw (031) node[blackvertex] {};
\draw (032) node[blackvertex] {};
\draw (033) node[blackvertex] {};
\draw (034) node[blackvertex] {};
\draw (035) node[blackvertex] {};
\draw (036) node[blackvertex] {};
\draw (037) node[blackvertex] {};
\draw (038) node[blackvertex] {};
\draw (039) node[blackvertex] {};
\draw (040) node[blackvertex] {};
\draw (041) node[blackvertex] {};

\draw (042) node[vertex] {$2_\#20$};

\draw (51) node[blackvertex] {};
\draw (52) node[blackvertex] {};
\draw (53) node[blackvertex] {};
\draw (54) node[blackvertex] {};
\draw (55) node[blackvertex] {};
\draw (56) node[blackvertex] {};
\draw (57) node[blackvertex] {};
\draw (58) node[blackvertex] {};
\draw (59) node[blackvertex] {};
\draw (61) node[blackvertex] {};
\draw (62) node[blackvertex] {};
\draw (63) node[blackvertex] {};
\draw (64) node[blackvertex] {};
\draw (65) node[blackvertex] {};
\draw (66) node[blackvertex] {};
\draw (67) node[blackvertex] {};
\draw (68) node[blackvertex] {};
\draw (69) node[blackvertex] {};

\end{tikzpicture}
\end{center}
\caption{Game trees of $2_\#0$ and $2_\#20$}\label{fig:ctrex}
\end{figure}


\end{sloppypar}

\end{document}